\documentclass[12pt]{amsart}
\usepackage{amsmath, amssymb, amsthm, mathrsfs, mathtools, amscd}
\usepackage[latin1]{inputenc}
\usepackage{hyperref}
\usepackage[all]{xy}

\newtheorem{theorem}{Theorem}[section]
\newtheorem{lemma}[theorem]{Lemma}
\newtheorem{corollary}[theorem]{Corollary}
\newtheorem{proposition}[theorem]{Proposition}
\newtheorem{definition}[theorem]{Definition}
\newtheorem{remark}[theorem]{Remark}
\newtheorem{example}[theorem]{Example}

\newcommand\cN{\mathcal{N}}
\newcommand\cM{\mathcal{M}}
\newcommand\PN{\mathcal{P}(\cN)}
\newcommand\PM{\mathcal{P}(\cM)}
\newcommand\hA{{\hat A}}
\newcommand\hB{\hat B}
\newcommand\hE{\hat E}
\newcommand\EA{E^{\hA}}
\newcommand\hEA{\hE^{\hA}}
\newcommand\bbR{\mathbb{R}}
\newcommand\bbN{\mathbb{N}}
\newcommand\ra{\rightarrow}
\newcommand\mt{\mapsto}
\newcommand\lra{\longrightarrow}
\newcommand\lmt{\longmapsto}
\newcommand\hra{\hookrightarrow}
\newcommand\bmeet{\bigwedge}
\newcommand\meet{\wedge}
\newcommand\bjoin{\bigvee}
\newcommand\join{\vee}
\newcommand\oA{o^{\hA}}
\newcommand\eR{\overline{\bbR}}
\newcommand\hP{{\hat P}}
\newcommand\hQ{\hat Q}
\newcommand\cH{\mathcal{H}}
\newcommand\Eo{E^o}
\newcommand\oE{o^E}
\renewcommand\sp{\operatorname{sp}}
\newcommand\sa{\operatorname{sa}}
\newcommand\mc[1]{\mathcal{#1}}
\newcommand\dito[2]{\delta^i({#1})_{#2}}
\newcommand\doto[2]{\delta^o({#1})_{#2}}

\newcommand\cA{\mathcal{A}}
\newcommand\eq[1]{(\ref{#1})}
\newcommand\PzN{\mathcal{P}_0(\cN)}

\newcommand\eA{e^{\hat A}}

\newcommand\ps[1]{\underline{#1}}
\newcommand\Sig{\ps{\Sigma}}
\newcommand\on[1]{\operatorname{#1}}

\newcommand\deo{\delta^o}
\newcommand\dei{\delta^i}

\newcommand\im{\on{im}}

\newcommand\cB{\mc B}
\newcommand\id{\on{id}}

\newcommand\aA{a^{\hA}}

\newcommand\FA{F^{\hA}}
\newcommand\hFA{\hF^{\hA}}

\newcommand\BH{\mc B(\cH)}
\newcommand\PH{\mc P(\cH)}
\newcommand\zA{z^{\hA}}
\newcommand\hF{\hat F}

\usepackage{color}
\definecolor{darkgreen}{rgb}{0,.66,0}
\usepackage[normalem]{ulem} 

\begin{document}
\title[Self-adjoint Operators as Functions I]{Self-adjoint Operators as Functions I:\\Lattices, Galois Connections, and the Spectral Order}

\author{Andreas D\"oring}
\address{Andreas D\"oring\newline
\indent Clarendon Laboratory\newline
\indent Department of Physics\newline%
\indent University of Oxford\newline
\indent Parks Road\newline
\indent Oxford OX1 3PU, UK}
\email{doering@atm.ox.ac.uk}
\author{Barry Dewitt}
\address{Barry Dewitt\newline
\indent Department of Engineering and Public Policy\newline
\indent Carnegie Mellon University\newline
\indent 5000 Forbes Avenue\newline
\indent Pittsburgh, PA, 15213\newline
\indent USA}
\email{barrydewitt@cmu.edu}
\date{December 5, 2013}

\begin{abstract}
Observables of a quantum system, described by self-adjoint operators in a von Neumann algebra or affiliated with it in the unbounded case, form a conditionally complete lattice when equipped with the spectral order. Using this order-theoretic structure, we develop a new perspective on quantum observables.

In this first paper (of two), we show that self-adjoint operators affiliated with a von Neumann algebra $\cN$ can equivalently be described as certain real-valued functions on the projection lattice $\PN$ of the algebra, which we call $q$-observable functions. Bounded self-adjoint operators correspond to $q$-observable functions with compact image on non-zero projections. These functions, originally defined in a similar form by de Groote in \cite{deG05}, are most naturally seen as adjoints (in the categorical sense) of spectral families. We show how they relate to the daseinisation mapping from the topos approach to quantum theory \cite{DI11}. Moreover, the $q$-observable functions form a conditionally complete lattice which is shown to be order-isomorphic to the lattice of self-adjoint operators with respect to the spectral order.

In a subsequent paper \cite{DoeDew12b}, we will give an interpretation of $q$-observable functions in terms of quantum probability theory, and using results from the topos approach to quantum theory, we will provide a joint sample space for all quantum observables.
\end{abstract}

\maketitle

\vspace{0.7cm}


\section{Introduction}			\label{Sec_Introd}
It is well-known that the self-adjoint operators representing observables of a quantum system form a real vector space. Yet, while adding self-adjoint operators and multiplying them by real numbers are mathematically natural operations, it is much less clear what these operations mean physically. For example, what physical interpretation is attached to the sum of position and momentum?

The real vector space of bounded self-adjoint operators describing the physical quantities of a quantum system can be regarded as the self-adjoint part of a complex operator algebra; in particular, we will focus on von Neumann algebras here. In order to include unbounded self-adjoint operators, we will consider operators affiliated with a given von Neumann algebra.

Crucially, we emphasise the order structure on the set of self-adjoint operators in (or affiliated with) a von Neumann algebra over the linear structure. Thus, we consider the partial order on projections, and more importantly, the spectral order \cite{Ols71,deG04} on self-adjoint operators. The latter, which may be less well-known, generalises the partial order on projections and differs from the usual linear order on self-adjoint operators if the algebra is nonabelian. We provide some background on order theory and von Neumann algebras in section \ref{Sec_Prelims}.

Let $\hA$ be a self-adjoint operator, contained in a von Neumann algebra $\cN$, or affiliated with it if $\hA$ is unbounded. Let $\PN$ denote the lattice of projections in $\cN$. The key observation is that the spectral family
\begin{equation}
			\EA:\bbR\ra\PN
\end{equation}
of $\hA$ is a meet-preserving map between meet-semilattices (technically, we will use the extended reals $\eR$ instead of $\bbR$ to obtain preservation of all meets). This implies that $\EA$ has a left adjoint $\oA$ in the sense of category theory such that $\EA$ and $\oA$ form a Galois connection. $\oA$ is a real-valued function on the projections in the von Neumann algebra $\cN$.

We call $\oA$ the \emph{$q$-observable function} of $\hA$. This function is determined uniquely by $\EA$, and hence by $\hA$, and conversely determines them uniquely. These functions were first considered by de Groote \cite{deG01}, and independently (as far we are aware) by Comman \cite{Com06}. Neither of these authors used the definition via Galois connections. In section \ref{Sec_DefAndBasicProperties}, we will show that the image of $\oA$ on non-zero projections equals the spectrum of $\hA$. Moreover, we will characterise $q$-observable functions abstractly without any reference to operators and show that they form a conditionally complete lattice isomorphic to the lattice of self-adjoint operators affiliated with $\cN$ with respect to the spectral order.

In section \ref{Sec_FurtherProperties}, we consider how $q$-observable functions behave under extension and restriction of their domain respectively codomain. This leads to a characterisation of the maps called \emph{outer} and \emph{inner daseinisation of self-adjoint operators}, which are approximations with respect to the spectral order and which play a key role in the topos approach to quantum theory \cite{DI08a,DI08b,DI08c,DI08d,DI11,Doe11}; see also \cite{HLS09,Wol10}. Moreover, we show that there exists a limited form of `functional calculus': for suitable monotone functions $f:\eR\ra\eR$, it holds that $o^{f(\hA)}=f(\oA)$.

In section \ref{Sec_qAntonymousFunctions}, it is shown that in addition to $q$-observable functions there exists a second sort of functions associated with self-adjoint operators affiliated with a von Neumann algebra $\cN$, called \emph{$q$-antonymous functions}. Multiplying the $q$-observable function $\oA$ of a self-adjoint operator $\hA$ by a negative real number gives the $q$-antonymous function of $-\hA$, that is, $a^{-\hA}=-\oA$.

Section \ref{Sec_PhysInterpret} provides physical interpretations of the mathematical results from previous sections and gives a short outlook on the second paper, ``Self-adjoint Operators as Functions II: Quantum Probability'' \cite{DoeDew12b}.

\section{Some mathematical preliminaries}			\label{Sec_Prelims}
In this section, we briefly present some basic definitions and results from order and lattice theory (subsection \ref{Subsec_OrderThBasics}) and from the theory of von Neumann algebras (subsection \ref{Subsec_VNABasics}) that will be used in the rest of the paper.


\subsection{Some order theory basics}			\label{Subsec_OrderThBasics}
Standard references on order theory are e.g. \cite{DavPri02,Bir67}. A \emph{partially ordered set (poset) $(P,\leq)$} is a set $P$ with a binary relation $\leq$, the \emph{order}, that is reflexive, antisymmetric and transitive. An element $\bot_P\in P$ such that $\bot_P\leq a$ for all elements $a\in P$ is called a \emph{bottom element}. If $P$ has a bottom element, it is necessarily unique. An element $\top_P$ such that $a\leq\top_P$ for all $a\in P$ is called a \emph{top element}. If $P$ has a top element, it is unique.

A \emph{meet-semilattice} is a poset $P$ such that any two elements $a,b$ have a \emph{meet (greatest lower bound)} $a\meet b$ in $P$, that is,
\begin{equation}
			\forall c\in P : c\leq a,b \quad\Longleftrightarrow\quad c\leq a\meet b.
\end{equation}
A meet-semilattice is called \emph{complete} if every family $(a_i)_{i\in I}$ of elements in $P$ has a greatest lower bound in $P$, denoted by $\bmeet_{i\in I} a_i$. In a complete meet-semilattice, the empty family has a meet $\bmeet\emptyset$, which is the top element of $P$. Also, the family containing all elements of $P$ has a meet $\bmeet_{a\in P} a$, which is the bottom element of $P$.

A \emph{join-semilattice} is a poset $P$ such that any two elements $a,b$ have a \emph{join (least upper bound)} $a\join b$ in $P$, that is,
\begin{equation}
			\forall c\in P : a,b\leq c \quad\Longleftrightarrow\quad a\join b\leq c.
\end{equation}
A join-semilattice is called \emph{complete} if every family $(a_i)_{i\in I}$ of elements in $P$ has a least upper bound in $P$, denoted by $\bjoin_{i\in I} a_i$. In a complete join-semilattice, the empty family has a join $\bjoin\emptyset$, which is the bottom element of $P$. Also, the family containing all elements of $P$ has a join $\bjoin_{a\in P} a$, which is the top element of $P$.

A poset $P$ that is both a meet-semilattice and a join-semilattice is called a \emph{lattice}. A lattice is \emph{complete} if it is complete as a meet- and a join-semilattice. If $P$ is a complete meet-semilattice, joins can be defined in $P$ by
\begin{equation}
			\forall (a_i)_{i\in I}\subseteq P: \bjoin_{i\in I} a_i := \bmeet\{b\in P \mid \forall i\in I: a_i\leq b\},
\end{equation}
that is, the least upper bound of the family $(a_i)_{i\in I}$ is the greatest lower bound of all $b\in P$ that are greater than all the $a_i$. Conversely, in a complete join-semilattice, meets can be defined in terms of joins.

A lattice $P$ is \emph{distributive} if, for all $a,b,c\in P$,
\begin{align}
			a\meet (b\join c) &= (a\meet b)\join(a\meet c),\\
			a\join (b\meet c) &= (a\join b)\meet(a\join c).
\end{align}
In fact, either of these conditions implies the other.

In categorical terms,\footnote{We will make clear how the order-theoretic constructions we use can be phrased naturally in the language of category theory, but the article is entirely understandable without knowledge of category theory.} a poset is a category $P$ with at most one arrow between any two objects $a,b$ (expressing the fact that $a\leq b$). Such a category $P$ is a meet-semilattice if binary products exist, and a join-semilattice if binary coproducts exist. Completeness corresponds to existence of all products resp. coproducts. The bottom element is a (necessarily unique) initial object, given by the empty join and the meet over all elements of $P$. The top element is a (necessarily unique) terminal object, given by the empty meet and the join over all elements of $P$.

\begin{example}			\label{Ex_Posets}
The natural numbers $\bbN$ with the usual order are a poset. There is a bottom element $\bot_{\bbN}=0$, but no top element. $(\bbN,\leq)$ is both a meet-semilattice, with meets given by minima, and a join-semilattice, with joins given by maxima, but $(\bbN,\leq)$ is not complete: neither the empty meet nor the join over all of $\bbN$ exists.

The real numbers $\bbR$ with their usual order are a poset with no bottom and no top element. $(\bbR,\leq)$ is a meet-semilattice, with minima as meets, and a join-semilattice, with maxima as joins, but it is not complete. 

But $\bbR$ is `almost' a complete meet-semilattice: any family $(r_i)_{i\in I}\subset\bbR$ that has a lower bound has a greatest lower bound, given by the \emph{infimum $\inf_{i\in I} r_i$}. This means that $\bbR$ is a \emph{conditionally} (or \emph{boundedly}) \emph{complete} meet-semilattice. It can be made into a complete meet-semilattice by adding a bottom element $-\infty$ (which is the meet over all elements of $\bbR$) and a top element $\infty$ (which is the empty meet). The \emph{extended reals}
\begin{equation}
			\eR=\{-\infty\}\cup\bbR\cup\{\infty\}
\end{equation}
form a complete meet-semilattice (with infima as meets) that will play an important role later on. Similarly, if a family $(r_i)_{i\in I}$ has an upper bound, it has a least upper bound, given by the \emph{supremum $\sup_{i\in I} r_i$}. This means that $\bbR$ is a conditionally (or boundedly) complete join-semilattice. The extended reals $\eR$ form a complete join-semilattice (with suprema as joins) and hence a complete lattice.

Both $(\bbN,\leq)$ and $(\bbR,\leq)$ are \emph{totally ordered}: for any pair $(a,b)$ of elements, either $a\leq b$ or $b\leq a$ (or both if $a=b$).
\end{example}

\begin{remark}
Of course, $\eR$ is a compactification of $\bbR$. Our emphasis in this article is on order-theoretic aspects, so we will not use topological arguments and structure explicitly (though many arguments could be phrased in topological terms).
\end{remark}

\begin{example}
Let $S$ be a set. The power set $P(S)$, consisting of all subsets of $S$ with inclusion as partial order, is a complete lattice. Meets are given by intersections, joins by unions. The bottom element in $(P(S),\subseteq)$ is the empty subset, the top element is $S$. Moreover, $(P(S),\subseteq)$ is a distributive lattice.

Let $\cN$ be a von Neumann algebra. The projections in $\cN$ form a complete lattice $\PN$ with respect to the order
\begin{equation}
			\forall \hP,\hQ\in\PN: \hP\leq\hQ :\quad\Longleftrightarrow\quad \hP\hQ=\hQ\hP=\hP.
\end{equation}
The lattice $(\PN,\leq)$ is distributive if and only if the von Neumann algebra $\cN$ is abelian.

The lattices $(P(S),\subseteq)$ and $(\PN,\leq)$ have additional structure: they both have \emph{complements}. In $(P(S),\subseteq)$, the complement of a subset $T\subseteq S$ is the set-theoretic complement $S\backslash T$. Together with unions and intersections, this makes $(P(S),\subseteq)$ into a complete Boolean algebra. In the projection lattice $(\PN,\leq)$, the (ortho)complement $\hP$ of a projection is $\hat 1-\hP$. Together with meets and joins, this makes $(\PN,\leq)$ into a complete orthomodular lattice, which is a complete Boolean lattice if and only if $\cN$ is abelian; see e.g. Prop. 2 in \cite{Red09}.
\end{example}

Let $(P,\leq)$ and $(Q,\sqsubseteq)$ be two posets. A map $f:P\ra Q$ is called \emph{monotone (order-preserving)} if
\begin{equation}
			\forall a,b\in P: a\leq b \quad\Longrightarrow\quad f(a)\sqsubseteq f(b).
\end{equation}
A pair of monotone maps $(f:P\ra Q,\;g:Q\ra P)$ is called a \emph{Galois connection} if
\begin{equation}			\label{Def_GaloisConnection}
			\forall a\in P\;\forall x\in Q: f(a)\sqsubseteq x \quad \Longleftrightarrow \quad a \leq g(x).
\end{equation}
Categorically, a monotone map is a functor from $P$ to $Q$, and a Galois connection is a pair of adjoint functors \cite{McL98}, where $f$ is the \emph{left adjoint} and $g$ the \emph{right adjoint}. Left adjoints preserve joins (least upper bounds), and right adjoints preserve meets (greatest lower bounds). If a functor $f:P\ra Q$ has a right adjoint $g:Q\ra P$, then $g$ is unique; likewise, if a functor $g:Q\ra P$ has a left adjoint $f:P\ra Q$, then $f$ is unique.

\begin{remark}			\label{Rem_UnitCounitComposite}
If $f:P\ra Q$ is a monotone map that has an adjoint $g:Q\ra P$, eq. \eq{Def_GaloisConnection} shows that $g$ is also monotone. If $(f:P\ra Q,\;g:Q\ra P)$ is a Galois connection, then it is easy to show that the composite map $g\circ f:P\ra P$, the \emph{unit of the adjunction}, is larger than $\id_P$, the identity map on $P$ (i.e., $g(f(a))\geq a$ for all $a\in P$). Similarly, $f\circ g:Q\ra Q$, the \emph{counit of the adjunction}, is smaller than $\id_Q$.

If we have two Galois connections $(f_1:P\ra Q,\;g_1:Q\ra P)$ and $(f_2:Q\ra R,\;g_2:R\ra Q)$, then the composite $(f_2\circ f_1:P\ra R,\;g_1\circ g_2:R\ra P)$ is a Galois connection, too, where $f_2\circ f_1$ is left adjoint to $g_1\circ g_2$.
\end{remark}


As we saw, complete meet-semilattices are also complete join-semilattices (and vice versa), and so are complete lattices. The following theorem is the \emph{adjoint functor theorem for posets} (see also Example 9.33 in \cite{Awo10}):
\begin{theorem}			\label{Thm_AdjFunctorThmForPosets}
Let $(P,\leq),(Q,\sqsubseteq)$ be complete join-semilattices. If $f:P\ra Q$ is a monotone map, then $f$ has a right adjoint $g:Q\ra P$ if and only if $f$ preserves all joins. The right adjoint $g$ is given by
\begin{align}
			g: Q &\lra P\\			\nonumber
			x &\lmt \bjoin\{a\in P \mid f(a)\sqsubseteq x\}.
\end{align}
The right adjoint $g$ is monotone and preserves all meets. 

Conversely, if $g:Q\ra P$ is a monotone map between complete meet-semilattices, then $g$ has a left adjoint $f:P\ra Q$ if and only if $g$ preserves all meets. The left adjoint $f$ preserves all joins and is given by
\begin{align}
			f: P &\lra Q\\			\nonumber
			a &\lmt \bmeet\{x\in Q \mid a\leq g(x)\}.
\end{align}
\end{theorem}

\begin{proof}
For the sake of completeness, we give a proof of this standard result. First, suppose that $f:P\ra Q$ has a right adjoint $g: Q \ra P$.  Let $(a_i)_{i\in I} \subseteq P$ be an arbitrary family of elements in $P$. Then, $a_i \leq \bigvee_{i\in I} a_i$ for all $i \in I$, so $f(a_i) \sqsubseteq f\left(\bigvee_{a_i\in I} a_i\right)$ for all $i \in I$ since $f$ is monotone, and therefore $\bigvee_{i\in I} f(a_i) \sqsubseteq f\left(\bigvee_{i\in I} a_i\right)$. Now, for all $i \in I$, we have $f(a_i) \sqsubseteq \bigvee_{i\in I} f(a_i)$. Since $g$ is right-adjoint to $f$, we have, for each $i \in I$, $a_i \leq g\left(\bigvee_{i\in I} f(a_i)\right)$ and so $\bigvee_{i\in I} a_i \leq g\left(\bigvee_{i\in I} f(a_i)\right)$. Therefore, since $f$ is left-adjoint to $g$, $f\left(\bigvee_{i\in I} a_i\right) \sqsubseteq \bigvee_{i\in I} f(a_i)$. Thus, we have $f\left(\bigvee_{i\in I} a_i\right) = \bigvee_{i\in I} f(a_i)$ for any family $(a_i)_{i\in I} \subseteq P$.

Now suppose that $f: P \ra Q$ preserves all joins.  Let $g: Q \ra P$ be defined by $g(x) = \bigvee\{a \in P \mid f(a)\sqsubseteq x\}$ for all $x \in Q$. Let $a_1 \in P$ such that $f(a_1) \sqsubseteq x$ for a given $x \in Q$. Then $a_1 \in \{a \in P \mid f(a) \sqsubseteq x\}$, so $g(x) = \bigvee\{a \in P \mid f(a)\sqsubseteq x\} \geq a_1$. Conversely, suppose that $a_1 \leq g(x)$.  Then, since $f$ is a monotone map and preserves joins,
\begin{align}
			f(a_1) \sqsubseteq f(g(x)) &= f\left(\bigvee\{a \in P \mid f(a) \sqsubseteq x\}\right)\\
			&= \bigvee \{f(a) \in Q \mid f(a) \sqsubseteq x\} \sqsubseteq x,
\end{align}
so $f(a_1) \sqsubseteq x$.  Therefore, $a \leq g(x) \Leftrightarrow f(a) \sqsubseteq x$ for all $a \in P$ and $x \in Q$, so by the definition of a Galois connection \eq{Def_GaloisConnection}, $g$ is right-adjoint to $f$.

The proof that a map $g:Q\ra P$ has a left adjoint $f:P\ra Q$ if and only if $g$ preserves all meets is completely analogous.
\end{proof}

The adjoint functor theorem for posets is a key result that will be used throughout.

\subsection{Von Neumann algebras, affiliated self-adjoint operators, and spectral families}			\label{Subsec_VNABasics}
Let $\cN$ be a von Neumann algebra, i.e., a subalgebra $\cN\subseteq\BH$ of the algebra of bounded operators on a complex Hilbert space $\cH$ that is closed in the weak (and strong) operator topology. $\BH$ itself is a von Neumann algebra. Standard references on von Neumann algebras and the various topologies on $\BH$ are \cite{KR83+86,Tak79+02}.

If $\cH$ is a complex Hilbert space, then $\PH$ denotes the lattice of projections onto closed subspaces of $\cH$. If $\cN$ is a von Neumann subalgebra of $\BH$, we will always assume that the unit element in $\cN$ is the identity (projection) $\hat 1$ on $\cH$. As mentioned above, the projections in $\cN$ form a complete orthomodular lattice $\PN$, and $\mc P(\BH)=\PH$. 

\begin{lemma}			\label{Lem_InclusionIsCOMLMorphism}
Let $\cN$ be a von Neumann algebra, and let $\cM\subset\cN$ be a von Neumann subalgebra such that the unit elements in $\cM$ and $\cN$ coincide. The inclusion map $i:\PM\ra\PN$, $\hP\mt\hP$, is a morphism of complete orthomodular lattices and hence has both a left adjoint $\deo_\cM:\PN\ra\PM$ and a right adjoint $\dei_\cM:\PN\ra\PM$.
\end{lemma}

\begin{proof}
The orthocomplement of a projection $\hP\in\PM$ is $\hat 1-\hP$. Since $i(\hat 1-\hP)=\hat 1-\hP=\hat 1-i(\hP)$, the inclusion $i$ preserves orthocomplements. Let $(\hP_i)_{i\in I}$ be a family of projections in $\PM$. Each projection corresponds to a closed subspace $S$ of $\cH$ such that $\hP(S)=S$ and $\hP(S^\perp)=0$ (see e.g. section 2.5 in \cite{KR83+86}). The meet $\bmeet_{i\in I}\hP_i$ is the projection onto the closed subspace given by the intersection of the closed subspaces that the $\hP_i$ project onto. This intersection is independent of whether the family $(\hP_i)_{i\in I}$ is considered to lie in $\PM$ or in $\PN$ (or in $\PH$), so $i$ preserves all meets. Moreover, $\bjoin_{i\in I}\hP_i=\hat 1-\bmeet_{i\in I}(\hat 1-\hP_i)$ by de Morgan's law, so $i$ preserves all joins, too.

$i$ preserves all meets, so by the adjoint functor theorem for posets (Thm. \ref{Thm_AdjFunctorThmForPosets}) it has a left adjoint, given concretely by $\deo_\cM:\PN\ra\PM$, $\hP\mt\deo(\hP)_{\cM}=\bmeet\{\hQ\in\PM \mid \hQ\geq\hP\}$. Since $i$ also preserves all joins, it has a right adjoint, too, given by $\dei_\cM:\PN\ra\PM$, $\hP\ra\dei(\hP)_\cM=\bjoin\{\hQ\in\PM \mid \hQ\leq\hP\}$.
\end{proof}


For the sake of completeness, we include some standard definitions.
\begin{definition}
Let $\cH$ be a complex Hilbert space, and let $\PH$ be its lattice of projections. A \emph{spectral family} is a map $E:\bbR \ra \PH$, $r\mt\hE_r$ such that
\begin{itemize}
	\item [(i)] for all $r,s\in\bbR$, if $r<s$, then $\hE_r\leq\hE_s$,
	\item [(ii)] $\bmeet_{r\in\bbR}\hE_r = \hat 0$ and $\bjoin_{r\in\bbR}\hE_r = \hat 1$.
\end{itemize}
If, moreover, it holds that
\begin{itemize}
	\item [(iii)] for all $r\in\bbR$, $\bmeet_{s>r}\hE_s=\hE_r$,
\end{itemize}
then $E:\bbR\ra\PH$ is called \emph{right-continuous}. Analogously, if for all $r\in\bbR$, $\bjoin_{s<r}\hE_s=\hE_r$, then $E$ is called \emph{left-continuous}. We will denote left-continuous spectral families by the letter $F$ instead of $E$.

Let $E$ be a right-continuous spectral family. If there is some $a\in\bbR$ such that $\hE_r=\hat 0$ for all $r<a$, then $E$ is \emph{bounded from below}. If there is some $b\in\bbR$ such that $\hE_r=\hat 1$ for all $r\geq b$, then $E$ is \emph{bounded from above}. $E$ is \emph{bounded} if it is bounded from below and from above.

Let $\cN$ be a von Neumann algebra on a Hilbert space $\cH$, and let $\PN$ be its projection lattice. If the image of a spectral family $E$ is in $\PN$, then we say that \emph{$E$ is in $\cN$}.
\end{definition}

A spectral family $E:\bbR\ra\PN$ in $\cN$ can be seen as a monotone function
\begin{align}
			E:\bbR &\lra \PN\\			\nonumber
			r &\lmt \hE_r.
\end{align}
We will use the notations $E:\bbR\ra\PN$ and $E=(\hE_r)_{r\in\bbR}$ for a spectral family interchangably. The following is the important spectral theorem, which relates spectral families and self-adjoint operators:
\begin{theorem}
Let $\hA$ be a self-adjoint operator on a Hilbert space $\cH$. There is a unique right-continuous spectral family $E:\bbR\ra\PH$ such that
\begin{align}			\label{Eq_SpecThm}
			\hA=\int_{-\infty}^{\infty} r\;d\hE_r
\end{align}
in the sense of norm-convergence of approximating Riemann sums. Conversely, every right-continuous spectral family $E:\bbR\ra\PH$ determines a unique self-adjoint operator on $\cH$ by equation \eq{Eq_SpecThm}. The right-continous spectral family corresponding to $\hA$ will be denoted $\EA$.
\end{theorem}

The spectral theorem holds analogously for left-continous spectral families.

\begin{remark}
A self-adjoint operator $\hA$ on $\cH$ is \emph{bounded (bounded from below/above)} if and only if its spectral family $\EA$ is bounded (bounded from below/above). $\hA$ is contained in the von Neumann algebra $\BH$ if and only if $\hA$ is bounded. If $\cN\subset\BH$ is some von Neumann (sub)algebra, then $\hA\in\cN$ if and only if the image of $\EA$ is in $\PN$ and $\hA$ is bounded.
\end{remark} 

\begin{definition}
If the image of $\EA$ is in $\PN$, but $\hA$ is not necessarily bounded, then we say that $\hA$ is \emph{affiliated with $\cN$}. The set of self-adjoint operators in a von Neumann algebra $\cN$ is denoted $\cN_{\sa}$. The set of self-adjoint operators affiliated with $\cN$ is denoted $SA(\cN)$ (and so $\cN_{\sa}\subset SA(\cN)$).
\end{definition}

In the following, we will simply speak of self-adjoint operators affiliated with a von Neumann algebra $\cN$, implicitly understanding that if such an operator $\hA$ is bounded then it lies in $\cN_{\sa}$. For more details on affiliated operators, see e.g. section 5.6 in \cite{KR83+86}. 

If $\hA$ is a bounded self-adjoint operator, then it is defined on all of $\cH$, while unbounded operators have domains of definition properly smaller than $\cH$. If $\hA$ is affililated with a von Neumann algebra $\cN\subseteq\BH$, then $\hA$ is a closed operator and is defined on a dense subset $\mc D(\hA)$ of $\cH$. We will not consider any questions relating to domains of definition of unbounded operators in this article, and we only need the concept in the following definition:
\begin{definition}			\label{Def_SpecOfSelfAdjOp}
Let $\hA$ be a self-adjoint operator affiliated with a von Neumann algebra $\cN$, and let $\mc D(\hA)\subseteq\cH$ be its (dense) domain of definition. The \emph{spectrum $\sp\hA$ of $\hA$} consists of those real numbers $s$ for which $\hA-s\hat 1$ is not a bijective mapping from $\mc D(\hA)$ to $\cH$.
\end{definition}
It is well known that $\sp\hA$ is non-empty and consists of those real numbers $s\in\bbR$ for which $\EA=(\EA_r)_{r\in\bbR}$ is not constant on any open neighbourhood $U$ of $s$. If $\hA$ is bounded, then $\sp\hA$ is a compact subset of $\bbR$.

We now define the spectral order, which will play a central role:
\begin{definition}
The \emph{spectral order} on the set $SA(\cN)$ of self-adjoint operators affiliated with a von Neumann algebra $\cN$ is defined by
\begin{equation}			\label{Def_SpecOrder}
			\forall \hA,\hB\in SA(\cN) : \hA\leq_s\hB\quad :\Longleftrightarrow \quad\forall r\in\bbR: \hEA_r\geq\hE^{\hB}_r,
\end{equation}
where $\EA=(\hEA_r)_{r\in\bbR}$ and $\hE^{\hat B}=(\hE^{\hat B}_r)_{r\in\bbR}$ are the right-continuous spectral families of $\hA$ resp. $\hB$, and where on the right-hand side, the usual order on projections is used.
\end{definition}

Of course, the spectral order can be restricted to $\cN_{\sa}$. This order was originally introduced by Olson \cite{Ols71}; see also \cite{deG04}. For some recent results on the spectral order for unbounded operators see \cite{PlaSto12}.

With respect to the spectral order, the set $SA(\cN)$ (resp. $\cN_{\sa}$) is a conditionally complete lattice, i.e., each family $(\hA_i)_{i\in I}$ in $SA(\cN)$ (resp. $\cN_{\sa}$) that has a lower bound has a greatest lower bound $\bmeet_{i\in I} \hA_i$ in $SA(\cN)$ (resp. $\cN_{\sa}$), and if the family has an upper bound, then it has a least upper bound $\bjoin_{i\in I} \hA_i$ in $SA(\cN)$ (resp. $\cN_{\sa}$). This is in marked contrast to the usual linear order that is given as
\begin{equation}
			\forall \hA,\hB\in SA(\cN) : \hA\leq\hB\quad :\Longleftrightarrow \quad\hB-\hA\text{ is positive.}
\end{equation}
As Kadison showed \cite{Kad51}, for a nonabelian von Neumann algebra $\cN$ the meet $\hA\meet\hB$ of two self-adjoint operators in $\cN$ exists if and only if $\hA$ and $\hB$ are comparable with respect to the linear order, i.e., if either $\hA\leq\hB$ or $\hB\leq\hA$. Hence, $\cN_{\sa}$ equipped with the linear order $\leq$ is very far from being a lattice; Kadison called $(\cN_{\sa},\leq)$ an \emph{anti-lattice}.

Some further facts about the spectral order (for proofs see \cite{Ols71,deG04}):
\begin{itemize}
	\item [(a)] On projections, the spectral order and the usual linear order coincide, so the spectral order generalises the partial order on projections.
	\item [(b)] On commuting operators, the spectral order and the usual linear order coincide, so for abelian von Neumann algebras, the linear order and the spectral order coincide.
	\item [(c)] For all $\hA,\hB\in SA(\cN)$, $\hA\leq_s\hB$ implies $\hA\leq\hB$, but not vice versa. In fact,  $\hA\leq_s\hB$ if and only if $\hA^n\leq\hB^n$ for all $n\in\bbN$.
	\item [(d)] The spectral order does not make $SA(\cN)$ (or $\cN_{\sa}$) into a vector lattice, that is, $\hA\leq_s\hB$ does not necessarily imply $\hA+\hat C\leq_s\hB+\hat C$ (unless $\cN$ is abelian).
\end{itemize}
Property (d) can be seen as an `incompatibility' between the linear structure and the order structure on $SA(\cN)$ (and on $\cN_{\sa}$). In this article, we will focus on the order structure provided by the spectral order.

\section{Definition and basic properties of $q$-observable functions}			\label{Sec_DefAndBasicProperties}


\begin{definition}
Let $\cN$ be a von Neumann algebra. A map
\begin{align}
			E:\eR &\lra \PN\\			\nonumber
			r &\lmt \hE_r
\end{align} 
such that
\begin{itemize}
	\item [(a)] $\hE_{-\infty}=\hat 0$ and $\hE_{\infty}=\hat 1$,
	\item [(b)] $E|_{\bbR}$ is a spectral family
\end{itemize}
is called an \emph{extended spectral family}. We will use both the notations $E:\eR\ra\PN$ and $(\hE_r)_{r\in\eR}$ for an extended spectral family.
\end{definition}

\begin{remark}
Obviously, every spectral family $E:\bbR\ra\PN$ determines a unique extended spectral family $E:\eR\ra\PN$ and vice versa, so there is a canonical bijection between spectral families (defined on $\bbR$) and extended spectral families (defined on $\eR$). If $E:\bbR\ra\PN$ is right-continuous, then its extension $E:\eR\ra\PN$ is also right-continuous. In the following, we will mostly work with extended right-continuous spectral families.
\end{remark}

The spectral theorem shows that there is a bijection between $SA(\cN)$, the set of self-adjoint operators affiliated with a given von Neumann algebra $\cN$, and the set $SF(\eR,\PN)$ of extended, right-continuous spectral families in $\PN$. 

\begin{lemma}
If we regard an extended right-continuous spectral family as a monotone function
\begin{align}
			E:\eR &\lra \PN\\			\nonumber
			r &\lmt \hE_r,
\end{align}
then $E$ preserves all meets, so it is a morphism of complete meet-semilattices. Conversely, any meet-preserving map $E:\eR\ra\PN$ with the properties 
\begin{itemize}
	\item [(a)] $E(-\infty)=\hat 0$,
	\item [(b)] $\bjoin_{r\in\bbR} E(r)=\hat 1$
\end{itemize}
determines an extended right-continuous spectral family.
\end{lemma}

\begin{proof}
Let $E:\eR\ra\PN$ be an extended right-continuous spectral family, and let $(r_i)_{i\in I}\subseteq\eR$ be an arbitrary family of elements of $\eR$. Then
\begin{align}
			E(\inf_{i\in I} r_i) = \hE_{\inf_{i\in I} r_i} = \bmeet_{s>\inf_{i\in I} r_i} \hE_s = \bmeet_{i\in I} \hE_{r_i},
\end{align}
where the second equality is due to right-continuity of $E$ and the third is due to monotonicity.

Conversely, if we have a meet-preserving map $E:\eR\ra\PN$ satisfying conditions (a) and (b), then clearly $E$ is monotone and (writing $\hE_r:=E(r)$) we have
\begin{equation}
			\forall r\in\bbR: \hE_r=\hE_{\inf_{s>r} s}=\bmeet_{s>r}\hE_s,
\end{equation}
so $E$ is right-continuous. Moreover,
\begin{equation}
			\bmeet_{r\in\bbR} \hE_r = \hE_{\inf_{r\in\bbR} r} = E(-\infty) = \hat 0,
\end{equation}
where we used meet-preservation of $E$ in the second step and assumption (a) in the last step. Also by meet-preservation, we obtain
\begin{equation}
			E(\infty)=E(\inf\emptyset)=\bmeet\emptyset=\hat 1,
\end{equation}
since the empty meet is a terminal object (in the sense of category theory, where meets are products), which in a poset is the top element. Together with assumption (b), this shows that $E:\eR\ra\PN$ is an extended right-continuous spectral family.
\end{proof}

By the adjoint functor theorem for posets, an extended spectral family $E$ has a left adjoint
\begin{equation}
			\oE:\PN \lra \eR
\end{equation}
that preserves arbitrary joins, i.e., for all families $(\hP_i)_{i\in I}\subseteq\PN$,
\begin{equation}
			\oE(\bjoin_{i\in I}\hP_i)=\sup_{i\in I} \oE(\hP_i).
\end{equation}
Note that `adjoint' is used here in the categorical sense, not the operator-theoretic one. (All operators that we consider are self-adjoint.)

\begin{definition}
If $\hA$ is a self-adjoint operator affiliated with a von Neumann algebra $\cN$ and $\EA=(\hEA_r)_{r\in\eR}$ is its extended right-continuous spectral family, then the left adjoint $o^{\EA}$ of $\EA$ is denoted $\oA:\PN\ra\eR$ and is called the \emph{$q$-observable function associated with the self-adjoint operator $\hA$}.
\end{definition}
The adjoint functor theorem provides the explicit form of the function $\oE$ adjoint to an extended right-continuous spectral family $E=(\hE_r)_{r\in\eR}$ in $\PN$: for all $\hP\in\PN$,
\begin{equation}
			\oE(\hP)=\inf\{r\in\eR \mid \hP\leq\hE_r\}.
\end{equation}
If $E=\EA$, then
\begin{equation}
			\oA(\hP)=\inf\{r\in\eR \mid \hP\leq\hEA_r\}.
\end{equation}
\begin{remark}
A very similar sort of functions associated with self-adjoint operators was defined by de Groote in \cite{deG01}, Prop. 6.2, but restricted to rank-$1$ projections (i.e., on projective Hilbert space $\mathbb P\cH$). He called these functions `observable functions', and we adapt this naming in order to show our indebtedness to him. In \cite{deG05}, Def. 2.7 and following arguments, de Groote arrived at a definition very close to the one above, considering join-preserving functions from the non-zero projections to the reals. Comman considers the functions defined above in \cite{Com06} and proves a number of their properties. It seems that his results were found independently of de Groote's earlier work.
\end{remark}

Yet, the observation that right-continuous spectral families $E$ (resp. $\EA$) have left adjoints $\oE$ (resp. $\oA$), which are exactly the $q$-observable functions, is new as far as we are aware. The definition via Galois connections is the most natural one, allows proving the properties of these functions in a direct and transparent way, and opens up new perspectives. For example, we already saw that preservation of arbitrary joins is an immediate consequence of the adjoint functor theorem for posets.

We have shown so far that for every each $\hA\in SA(\cN)$,
\begin{equation}
			(\oA,\EA)
\end{equation}
is a Galois connection between the complete lattices $\PN$ and $\eR$. Also note that for all non-zero projections $\hP\in\PN$,
\begin{equation}
			\oA(\hP)=\inf\{r\in\eR \mid \hP\leq\hEA_r\}>-\infty.
\end{equation}
If $\hA$ is unbounded from below, then the image of $\oA$ on non-zero projections is unbounded from below: let $r_0$ be an element in the spectrum of $\hA$, then $\oA(\hEA_{r_0})=r_0$. (This fact will also be used below in the proof of Lemma \ref{Lem_im(oA)=spA}.) For the same reason, if $\hA$ is unbounded from above, then the image of $\oA$ is unbounded from above. In fact, there exist projections $\hP\in\PN$ such that
\begin{equation}
			\oA(\hP)=\inf\{r\in\eR \mid \hP\leq\hEA_r\}=\infty.
\end{equation}
For example, the identity $\hat 1$ is such a projection: if $\hA$ is unbounded from above, then $\hat 1$ is not in the usual spectral family of $\hA$ defined over $\bbR$.

It always holds that $\oA(\hat 0)=-\infty$, whether $\hA$ is bounded or not.

Of course, for any self-adjoint operator $\hA$ affiliated with $\cN$, there exist projections $\hP$ such that $\oA(\hP)$ is some finite real number: take $\hP=\hEA_{r_0}$ for some $r_0\in\sp\hA$ (this is always possible, since the spectrum of a self-adjoint operator is not empty), then $\oA(\hEA_{r_0})=\inf\{r\in\eR \mid \hEA_{r}\geq\hEA_{r_0}\}=r_0$.

We will now characterise those functions $o:\PN\ra\eR$ that determine extended spectral families by means of a Galois connection, and hence determine self-adjoint operators affiliated with the von Neumann algebra $\cN$.

\begin{definition}			\label{Def_AbsObsFct}
A \emph{weak $q$-observable function} is a join-preserving function $o:\PN\ra\eR$ such that
\begin{itemize}
	\item [(a)] $o(\hP)>-\infty$ for all $\hP>\hat 0$.
\end{itemize}
An \emph{abstract $q$-observable function on $\PN$} is a weak $q$-observable function such that
\begin{itemize}
	\item [(b)] there is a family $(\hP_i)_{i\in I}\subseteq\PN$ with $\bjoin_{i\in I}\hP_i=\hat 1$ such that $o(\hP_i)\lneq\infty$ for all $i\in I$.
\end{itemize}
The set of abstract $q$-observable functions on $\PN$ is denoted $QO(\PN,\eR)$.
\end{definition}

Let $o$ be a weak or an abstract $q$-observable function. Note that join-preservation implies
\begin{equation}			\label{Eq_o(0)=-infty}
			o(\hat 0)=o(\bjoin\emptyset)=\sup\emptyset=-\infty.
\end{equation}

\begin{proposition}			\label{Prop_ObsFcts=SpecFams}
Let $\cN$ be a von Neumann algebra. There is a bijection between the set $SF(\eR,\PN)$ of extended right-continuous spectral families in $\PN$ and the set $QO(\PN,\eR)$ of abstract $q$-observable functions on $\PN$. Concretely, each abstract $q$-observable function $o$ has a right adjoint $\Eo$, which is an extended right-continuous spectral family, and each extended right-continuous spectral family $E$ has a left adjoint $\oE$, which is an abstract $q$-observable function. Moreover, $E^{\oE}=E$ and $o^{\Eo}=o$ for all $E\in SF(\eR,\PN)$ and all $o\in QO(\PN,\eR)$.
\end{proposition}

\begin{proof}
Let $o:\PN\ra\eR$ be an abstract $q$-observable function. The adjoint functor theorem for posets shows that $o$ has a right adjoint $\Eo$ that preserves arbitrary meets (greatest lower bounds). In particular, for all $r\in\eR$, we have $\Eo(r)=\Eo(\inf\{s\in\eR \mid r<s\})=\bmeet_{s>r}\Eo(s)$, so $\Eo$ is right-continuous. Join-preservation of $o$ implies $o(\hat 0)=-\infty$, see \eq{Eq_o(0)=-infty}. From this and $o(\hP)>\hat 0$ for all $\hP>0$ (which is condition (a) in Def. \ref{Def_AbsObsFct}), we obtain
\begin{equation}			\label{Eq_Eo(-infty)=0}
			\Eo(-\infty)=\bjoin\{\hP\in\PN \mid o(\hP)\leq-\infty\}=\hat 0.
\end{equation}

It remains to show that $\bjoin_{r\in\bbR}\Eo(r)=\hat 1$. We have
\begin{align}
			&\bjoin_{r\in\bbR}\Eo(r) \stackrel{\text{Def. \ref{Def_AbsObsFct}, Cond. (b)}}{\geq} \bjoin_{r\in\{o(\hP_i) \mid i\in I\}}\Eo(r) = \bjoin_{i\in I}\Eo(o(\hP_i)) \geq \bjoin_{i\in I}\hP_i = \hat 1,
\end{align}
where in the last step we used the fact that for the unit of the adjunction, we have $\Eo\circ o\geq\id_{\PN}$ (see Rem. \ref{Rem_UnitCounitComposite}). Note that $\bjoin_{r\in\bbR}\Eo(r)=\hat 1$ is trivial if $o(\hat 1)\in\bbR$, since $o(\hat 1)$ is the maximum of the image of $o$ due to monotonicity, and so for all $t\in\eR$ with $t\geq o(\hat 1)$, we have $\Eo(t)=\bjoin\{\hP\in\PN \mid o(\hP)\leq t\}=\hat 1$. Also, if $o(\hat 1)\in\bbR$, condition (b) in Def. \ref{Def_AbsObsFct} is fulfilled for the family $\{\hat 1\}$.

Conversely, let $E:\eR\ra\PN$ be an extended right-continuous spectral family. Its left adjoint $\oE$ is the (concrete) $q$-observable function of $\hA^E$, the self-adjoint operator affiliated with $\cN$ that is determined by $E$. Being a left adjoint, $\oE$ preserves all joins. If $\hP>\hat 0$, then $\oE(\hP)=\inf\{r\in\eR \mid \hP\leq E(r)\}>-\infty$, so $\oE$ has property (a) in Def. \ref{Def_AbsObsFct}. Moreover, $E|_{\bbR}=(E(r))_{r\in\bbR}$ is a family of projections with $\bjoin_{r\in\bbR}E(r)=\hat 1$, and we have
\begin{equation}
			\forall r\in\bbR: \oE(E(r))\leq r\lneq\infty,
\end{equation}
because for the counit $\oE\circ E$ of the adjunction, it holds that $\oE\circ E\leq\id_{\eR}$ (see Rem. \ref{Rem_UnitCounitComposite}). Hence, condition (b) in Def. \ref{Def_AbsObsFct} is fulfilled for the family $(E(r))_{r\in\bbR}$. In fact, one can always use a countable family, e.g. pick $(E(n))_{n\in\bbN}$. If $\oE(\hat 1)\in\bbR$, then the trivial family $\{\hat 1\}$ can be used.

The fact that we are using an adjunction (that is, a Galois connection between posets) implies that $E^{\oE}=E$ and $o^{\Eo}=o$ for all $E\in SF(\eR,\PN)$ and all $o\in QO(\PN,\eR)$, since adjoints are unique.
\end{proof}

This leads us to our representation of self-adjoint operators as real-valued functions:

\begin{theorem}			\label{Thm_SA(N)=QO}
Let $\cN$ be a von Neumann algebra. There is a bijection between the set $SA(\cN)$ of self-adjoint operators affiliated with $\cN$ and the set $QO(\PN,\eR)$ of abstract $q$-observable functions on $\PN$.
\end{theorem}

\begin{proof}
Each $\hA\in SA(\cN)$ determines a unique extended right-continuous spectral family $\EA$ in $\PN$ by the spectral theorem. By Prop. \ref{Prop_ObsFcts=SpecFams}, $\EA$ determines a unique abstract $q$-observable function $\oA=o^{\EA}$. Conversely, each abstract $q$-observable function $o$ determines a unique extended right-continuous spectral family $\Eo$ and hence a unique self-adjoint operator $\hA^{\Eo}$.
\end{proof}

In other words, every abstract $q$-observable function $o$ is the $q$-observable function $\oA$ of some self-adjoint operator $\hA$ affiliated with $\cN$, and every self-adjoint operator $\hA$ affiliated with $\cN$ also determines an abstract $q$-observable function $\oA$. Hence, we will speak simply of $q$-observable functions in the following.

Summing up, we have three bijections:
\begin{align}
			&SA(\cN)\cong SF(\eR,\PN),\quad \hA\mt\EA,\quad \EA\mt\int_{-\infty}^{\infty} r\;d\hEA_r;\\
			&SF(\eR,\PN)\cong QO(\PN,\eR),\quad E\mt\oE,\quad o\mt\Eo;\\
			&SA(\cN)\cong QO(\PN,\eR),\quad \hA\mt o^{\hA},\quad o\mt\int_{-\infty}^{\infty} r\;d\Eo(r).
\end{align}

\subsection{Representation of the lattice of self-adjoint operators} 
We now consider more concretely how the spectral order relates to our constructions. In particular, the spectral order on self-adjoint operators corresponds to the pointwise order on the associated $q$-observable functions:
\begin{proposition}			\label{Prop_OrderIso}
Let $(SA(\cN),\leq_s)$ be the poset of self-adjoint operators affiliated with $\cN$, equipped with the spectral order $\leq_s$, and let \mbox{$(QO(\PN,\eR),\leq)$} be the poset of $q$-observable functions, equipped with the pointwise order. The map
\begin{align}			\label{E_SANleqs=QOleq}
			\phi: (SA(\cN),\leq_s) &\lra (QO(\PN,\eR),\leq)\\			\nonumber
			\hA &\lmt \oA
\end{align}
is an order-isomorphism of conditionally complete lattices.
\end{proposition}

\begin{proof}
Let $\hA,\hB$ be two self-adjoint operators affiliated with a von Neumann algebra $\cN$. Then
\begin{equation}
			\hA\leq_s\hB \quad\Longleftrightarrow\quad \forall r\in\eR: \hEA_r\geq\hE^{\hB}_r,
\end{equation}
which implies, for all $\hP\in\PN$,
\begin{equation}
			\oA(\hP)=\inf\{r\in\eR \mid \hEA_r\geq\hP\} \leq \inf\{t\in\eR \mid \hE^{\hB}_t\geq\hP\}=o^{\hat B}(\hP).
\end{equation}
Conversely, if $\oA\leq o^{\hB}$, then
\begin{align}
			\forall r\in\eR: \hEA_r &= \bjoin\{\hP\in\PN \mid r\geq\oA(\hP)\}\\
			&\geq\bjoin\{\hQ\in\PN \mid r\geq o^{\hB}(\hQ)\}\\
			&=\hE^{\hB}_r,
\end{align}
so $\hA\leq_s\hB$.
\end{proof}
Hence, we can represent the set $(SA(\cN),\leq_s)$ of self-adjoint operators affiliated with the von Neumann algebra $\cN$, equipped with the spectral order, faithfully by the set of real-valued functions in \mbox{$(QO(\PN,\eR),\leq)$}, partially ordered under the pointwise order.

Let $(SF(\eR,\PN),\leq_i)$ be the poset of extended spectral families in $\PN$, equipped with the inverse pointwise order, that is,
\begin{equation}
			\EA\leq_i\hE^{\hB} \quad :\Longleftrightarrow \quad (\forall r\in\eR: \hEA_r\geq\hE^{\hB}_r).
\end{equation}
Then $(SF(\eR,\PN),\leq_i)$ is order-isomorphic to $(SA(\cN),\leq_s)$ by definition of the spectral order, and one easily shows:
\begin{lemma}
The map
\begin{align}
			\phi': (QO(\PN,\eR),\leq) &\lra (SF(\eR,\PN),\leq_i)\\
			o &\lmt \Eo
\end{align}
is an order-isomorphism of conditionally complete lattices.
\end{lemma}

%
%
%

\subsection{$q$-observable functions and spectra of self-adjoint operators}
Let $\hA$ be a self-adjoint operator affiliated with a von Neumann algebra $\cN$, and let $\EA:\eR\ra\PN$ be its extended spectral family.

\begin{remark}			\label{Rem_inftyInSpec}
Since extended spectral families are defined over the extended reals $\eR$, it makes sense to regard the spectrum of a self-adjoint operator $\hA$ as a subset of $\eR$. We include $-\infty$ in the spectrum of $\hA$ if $\hA$ is unbounded from below, and include $\infty$ in the spectrum if $\hA$ is unbounded from above. Since extended spectral families are right-continuous at $-\infty$ and left-continuous at $\infty$, the `spectral values' $-\infty$ and $\infty$ are never isolated points of the spectrum, i.e., they are never eigenvalues.
\end{remark}

Let $\hA$ be a self-adjoint operator affiliated with a von Neumann algebra $\cN$, and let $\oA$ be the corresponding $q$-observable function. There is a straightforward relation between the spectrum of $\hA$ and the image of $\oA$:
\begin{lemma}			\label{Lem_im(oA)=spA}
Let $\PzN$ denote the non-zero projections in a von Neumann algebra $\cN$, and let $\oA\in QO(\PN,\eR)$ be the $q$-observable function corresponding to a self-adjoint operator $\hA\in SA(\cN)$. Then
\begin{equation}
			\oA(\PzN)=\sp\hA,
\end{equation}
i.e., the image of $\oA$ on the non-zero projections is equal to the spectrum of the operator. If $\hA$ is unbounded from above, then $\infty$ is in $\oA(\PzN)$.
\end{lemma}
\begin{proof}
Let $r\in\sp\hA$, then
\begin{align}
			\oA(\hEA_r)=\inf\{s\in\eR \mid \hEA_s\geq\hEA_r\},
\end{align}
which equals $r$ by right-continuity of $\EA$ and the fact that $\sp\hA$ consists of those real numbers $r$ for which $\EA$ is not constant on any open neighbourhood of $r$ (cf. paragraph after Def. \ref{Def_SpecOfSelfAdjOp}). Conversely, if $r$ is in the image of $\oA$, then $\hEA_t<\hEA_r$ for all $t<r$, so $\EA$ is not constant on any neighbourhood of $r$, hence $r\in\sp A$.

If $\hA$ is unbounded from above, then $\hEA_s<\hat 1$ for all $s\in\bbR$ and $\hEA_{\infty}=\hat 1$, so $\oA(\hat 1)=\infty$.
\end{proof}
In Remark 2.13 in \cite{deG05}, de Groote had shown (for bounded operators) that the image of $\oA$ on non-zero projections is dense in the spectrum of $\hA$. We now see that in fact these two sets are equal.

Of course, if one considers all projections including $\hat 0$, then
\begin{equation}
			\oA(\PN)=\sp\hA\cup\{-\infty\}.
\end{equation}

\begin{lemma}			\label{Lem_infIsmin}
Let $\hA$ be a self-adjoint operator affiliated with $\cN$, and let $\oA:\PN\ra\eR$ be its $q$-observable function. Then
\begin{equation}
			\forall \hP\in\PzN: \oA(\hP)=\inf\{r\in\eR \mid \hP\leq\hEA_r\}=\min\{r\in\eR \mid \hP\leq\hEA_r\}.
\end{equation}
\end{lemma}

\begin{proof}
This follows directly from meet-preservation of $\EA:\eR\ra\PN$.
\end{proof}

Since a bounded self-adjoint operator has a compact spectrum, we have:
\begin{corollary}			\label{Cor_BoundedCompactImage}
If $\hA$ is a bounded self-adjoint operator in $\cN$, then the image $\oA(\PzN)$ of the corresponding $q$-observable function $\oA$ is compact.
\end{corollary}

\begin{proposition}
Let $\cN$ be a von Neumann algebra. There is a bijection between $\cN_{\sa}$, the set of self-adjoint operators in $\cN$, and $QO^c(\PN,\eR)$, the $q$-observable functions with compact image on $\mc P_0(\cN)$.
\end{proposition}

\begin{proof}
This follows immediately from Thm. \ref{Thm_SA(N)=QO}, the fact that all self-adjoint operators in $\cN$ are bounded, and Cor. \ref{Cor_BoundedCompactImage}. 
\end{proof}

\section{Further properties of $q$-observable functions}			\label{Sec_FurtherProperties}
\subsection{Restricting the codomain} Let $\hA$ be a (bounded) self-adjoint operator in a von Neumann algebra $\cN$. Then the compact set $\sp\hA$ is a complete sublattice of $\eR$, that is, arbitrary infima and suprema exist in $\sp\hA$ and coincide with those in $\eR$. Hence, we can describe $\hA$ equivalently by a join-preserving function
\begin{align}
			\oA:\PN &\lra \sp\hA\cup\{-\infty\}\\			\nonumber
			\hP &\lmt \inf\{r\in\sp\hA\cup\{-\infty\} \mid \hEA_r\geq\hP\},
\end{align}
given by restricting the codomain of $\oA:\PN\ra\eR$. Note that $\oA(\PzN)=\sp\hA$; only the zero projection is mapped to $-\infty$. More generally, if $C$ is any compact subset of $\bbR$ that contains $\sp\hA$, then $C\cup\{-\infty\}$ can serve as the codomain of $\oA$. Note that if we consider all self-adjoint operators in $\cN$ (which by definition are bounded), then $\bbR\cup\{-\infty\}\subset\eR$ is the smallest common codomain of the associated $q$-observable functions.

\subsection{Restricting the domain -- outer daseinisation of self-adjoint operators}			\label{Subsec_Restrict}
Let $\hA$ be self-adjoint and affiliated with $\cN$, and let $\cM\subset\cN$ be a von Neumann subalgebra such that the unit elements in $\cN$ and $\cM$ coincide. In general, $\hA$ is not affiliated with $\cM$. In this subsection, we will consider an approximation of $\hA$ by a (generalised) self-adjoint operator $\doto{\hA}{\cM}$ in $\cM$, where `generalised' means that $\doto{\hA}{\cM}$ may have $\infty$ as an isolated point (i.e., an `eigenvalue') of its spectrum. We will show that for a self-adjoint operator $\hA$ that is bounded from above, $\doto{\hA}{\cM}$ is always a proper self-adjoint operator affiliated with $\cM$.

The (generalised) operator $\doto{\hA}{\cM}$ is called the \emph{outer daseinisation of $\hA$ to $\cM$} and plays an important role in the topos approach to quantum theory. Our main focus is on the $q$-observable function $o^{\doto{\hA}{\cM}}$ of $\doto{\hA}{\cM}$, and how it arises from Galois connections.

Recall from Lemma \ref{Lem_InclusionIsCOMLMorphism} that the inclusion $i:\PM\hra\PN$ is a morphism of complete orthomodular lattices that has a left adjoint $\deo_{\cM}:\PN\ra\PM$ and a right adjoint $\dei_{\cM}:\PN\ra\PM$. Let $\hA$ be a self-adjoint operator affiliated with $\cN$, and let $\EA:\eR\ra\PN$ be the extended right-continuous spectral family of $\hA$. We define a new map $E^{\doto{\hA}{\cM}}:\eR\ra\PM$ by
\begin{align}
			E^{\doto{\hA}{\cM}}:=\dei_{\cM}\circ\EA.
\end{align}
The idea is to approximate $\EA$ in $\PM$ by taking, for each $r\in\eR$, the largest projection in $\cM$ that is dominated by $\hEA_r$. It is easy to see that $E^{\doto{\hA}{\cM}}$ is a monotone, right-continuous map from $\eR$ to $\PM$ with $\bmeet_{r\in\bbR}\hE^{\doto{\hA}{\cM}}_r=\hat 0$ and $\hE^{\doto{\hA}{\cM}}_{\infty}=\hat 1$. Yet, in general, we need not have \begin{align}			\label{Eq_LeftContinAtInfty}
			\bjoin_{r\in\bbR}\hE^{\doto{\hA}{\cM}}_r=\hat 1,
\end{align}
since $\dei_{\cM}(\hEA_r)\leq\hEA_r$ for each $r\in\bbR$. We call $\hE^{\doto{\hA}{\cM}}$ a \emph{weak right-continuous} extended spectral family. It is a (proper) extended spectral family if and only if eq. \eq{Eq_LeftContinAtInfty} holds.  It is straightforward to show that the left adjoint $o^{\doto{\hA}{\cM}}$ of $E^{\doto{\hA}{\cM}}$ is a weak $q$-observable function (cf. Def. \ref{Def_AbsObsFct}). In any case, we can form
\begin{equation}
			\doto{\hA}{\cM}:=\int_{-\infty}^{\infty} r\;d(E^{\doto{\hA}{\cM}})=\int_{-\infty}^{\infty} r\;d(\dito{\hEA_r}{\cM}),
\end{equation}
which is called the \emph{outer daseinisation of $\hA$ to $\cM$}. $\doto{\hA}{\cM}$ is a self-adjoint operator affiliated with $\cM$ if and only if $\hE^{\doto{\hA}{\cM}}=\dei_{\cM}\circ\EA$ is a proper extended spectral family, that is, if eq. \eq{Eq_LeftContinAtInfty} holds. In this case, by construction we have
\begin{align}			\label{Eq_dotoAsMeet}
			\doto{\hA}{\cM}=\bmeet\{\hB\in SA(\cM) \mid \hB\geq_s\hA\},
\end{align}
where the meet is taken with respect to the spectral order on $SA(\cM)$. That is, $\doto{\hA}{\cM}$ is the smallest self-adjoint operator in $SA(\cM)$ that dominates $\hA$ with respect to the spectral order. If eq. \eq{Eq_LeftContinAtInfty} does not hold, then the meet in \eq{Eq_dotoAsMeet} is not defined in the conditionally complete lattice $SA(\cM)$.\footnote{One could consider a lattice completion of $SA(\cM)$, but this would lead us too far from the goals of this paper.} In this case, $\doto{\hA}{\cM}$ can be understood as a `self-adjoint operator with eigenvalue $\infty$', where $\hat 1-\bjoin_{r<\infty}\dito{\hEA_r}{\cM}$ is the projection onto the `eigenspace' of $\infty$.

\begin{lemma}			\label{Lem_BoundedFromAboveAndBounded}
If $\hA\in SA(\cN)$ is bounded from above, then $\hE^{\doto{\hA}{\cM}}$ is an extended right-continuous spectral family, and hence $\doto{\hA}{\cM}$ is a (proper) self-adjoint operator affiliated with $\cM$. If $\hA$ is bounded, then $\doto{\hA}{\cM}$ is bounded, too.
\end{lemma}

\begin{proof}
If $\hA$ is bounded from above, there is an $r_0\in\bbR$ such that $\hEA_{r_0}=\hat 1$, so $\hE^{\doto{\hA}{\cM}}_{r_0}=\hat 1$ and $\bjoin_{r\in\bbR}\hE^{\doto{\hA}{\cM}}_r=\hat 1$. Note that $r_0$ is an upper bound for $\sp(\doto{\hA}{\cM})$. If $\hA$ is also bounded from below, then there exists some $s_0\in\bbR$ such that $\hEA_{s_0}=\hat 0$, so $\hE^{\doto{\hA}{\cM}}_{s_0}=\hat 0$, and $s_0$ is a lower bound for $\sp(\doto{\hA}{\cM})$.
\end{proof}

We remark that $\hA$ being bounded from above is a sufficient, but not a necessary condition for $E^{\doto{\hA}{\cM}}$ to be an extended spectral family.

Now observe that $E^{\doto{\hA}{\cM}}=\dei_{\cM}\circ\EA:\eR\ra\PM$ is the composite of two right adjoints, so it is a right adjoint itself (cf. Rem. \ref{Rem_UnitCounitComposite}). The corresponding left adjoint is
\begin{align}
			o^{\doto{\hA}{\cM}}=\oA_{\cN}\circ i:\PM\lra\eR,
\end{align}
This is a weak $q$-observable function, and if $\hE^{\doto{\hA}{\cM}}$ is an extended right-continuous spectral family, then $o^{\doto{\hA}{\cM}}$ is a $q$-observable function. Since $i:\PM\ra\PN$ is an inclusion, we have $\oA_{\cN}\circ i=\oA_{\cN}|_{\PM}$. Summing up, we obtain:

\begin{proposition}			\label{Prop_OuterDasAndqObsFcts}
Let $\cN$ be a von Neumann algebra, let $\hA\in SA(\cM)$, and let $\cM\subset\cN$ be a von Neumann subalgebra such that the unit elements in $\cM$ and $\cN$ coincide. The weak $q$-observable function of $\doto{\hA}{\cM}$, the outer daseinisation of $\hA$ to $\cM$, is given by $o^{\doto{\hA}{\cM}}=\oA_{\cN}\circ i=\oA|_{\PM}$. If $\hA$ is bounded from above, then $\doto{\hA}{\cM}=\bmeet\{\hB\in SA(\cM) \mid \hB\geq_s\hA\}\in SA(\cM)$, and $o^{\doto{\hA}{\cM}}$ is a (proper) $q$-observable function.
\end{proposition}

This means that the weak $q$-observable function $o^{\doto{\hA}{\cM}}$ of the outer daseinisation of $\oA$ to \emph{any} von Neumann subalgebra $\cM$ is simply given by restriction of $\oA$ to the smaller domain $\PM$. Hence, $\oA$ encodes all outer daseinisations to subalgebras in a uniform way.

\subsection{Inner daseinisation}			\label{Subsec_InnerDasAndzAFcts}
We will show that there is also an \emph{inner daseinisation $\dito{\hA}{\cM}$} of a self-adjoint operator $\hA\in SA(\cN)$ to a von Neumann subalgebra $\cM\subset\cN$, and that it arises from the restriction of some function from $\PN$ to $\eR$, too. 

Let $\hA\in SA(\cN)$, and let $\FA$ be its \emph{left}-continuous extended spectral family. Note that $\FA:\eR\ra\PN$ preserves all joins, so it has a right adjoint $\zA:\PN\ra\eR$. Since $\zA$ determines $\FA$ uniquely, it also determines $\hA$ uniquely by the spectral theorem (which holds for left-continuous spectral families, too). The composite
\begin{align}
			F^{\dito{\hA}{\cM}}:=\deo_{\cM}\circ\FA:\eR\lra\PM
\end{align}
preserves joins, and hence is left-continuous, because both $\FA$ and $\deo_{\cM}$ preserve joins.\footnote{In contrast, $\deo_{\cM}\circ\EA$ is not right-continuous -- this is the reason for using the left-continuous spectral family $\FA$ instead of the right-continuous one, $\EA$.} The idea is to approximate $\FA$ in $\PM$ by taking, for each $r\in\eR$, the smallest projection in $\PM$ that dominates $\hFA_r$. It is easy to see that $\hF^{\dito{\hA}{\cM}}_{-\infty}=\hat 0$ and $\bjoin_{r\in\bbR}\hF^{\dito{\hA}{\cM}}_r=\hat 1$. Yet, in general it need not hold that
\begin{align}			\label{Eq_RightContinAtNinusInfty}
			\bmeet_{r\in\bbR}\hF^{\dito{\hA}{\cM}}_r=\hat 0,
\end{align}
since $\deo_{\cM}(\hFA_r)\geq\hFA_r$ for all $r\in\bbR$. We call $F^{\dito{\hA}{\cM}}$ a \emph{weak left-continuous} extended spectral family. It is a (proper) left-continuous extended spectral family if eq. \eq{Eq_RightContinAtNinusInfty} holds. We define the \emph{inner daseinisation of $\hA$ to $\cM$} by
\begin{align}
			\dito{\hA}{\cM}:=\int_{-\infty}^\infty r\;d(F^{\dito{\hA}{\cM}})=\int_{-\infty}^\infty r\;d(\doto{\hFA_r}{\cM}).
\end{align}
This is a self-adjoint operator affiliated with $\cM$ if and only if $\deo_{\cM}\circ\FA$ is a (proper) left-continuous extended spectral family, i.e., if eq. \eq{Eq_RightContinAtNinusInfty} holds. Otherwise, $\dito{\hA}{\cM}$ can be thought of as a `self-adjoint operator with eigenvalue $-\infty$', where $\bmeet_{r\in\bbR}\hF^{\dito{\hA}{\cM}}_r$ is the projection onto the `eigenspace' of $-\infty$. In analogy to Lemma \ref{Lem_BoundedFromAboveAndBounded}, one can show:

\begin{lemma}			\label{Lem_BoundedFromBelowAndBounded}
If $\hA\in SA(\cN)$ is bounded from below, then $F^{\dito{\hA}{\cM}}$ is an extended left-continuous spectral family, and hence $\dito{\hA}{\cM}$ is a (proper) self-adjoint operator affiliated with $\cM$. If $\hA$ is bounded, then $\dito{\hA}{\cM}$ is bounded, too.
\end{lemma}

$F^{\dito{\hA}{\cM}}=\deo_{\cM}\circ\FA$ is the composite of two left adjoints, and hence is a left adjoint itself. The corresponding right adjoint is
\begin{align}
			z^{\dito{\hA}{\cM}}=\zA\circ i: \PM \lra \eR.
\end{align}
In analogy to Prop. \ref{Prop_OuterDasAndqObsFcts}, we obtain:
\begin{proposition}			\label{Prop_InnerDasAndzFcts}
Let $\cN$ be a von Neumann algebra, let $\hA\in SA(\cM)$, and let $\cM\subset\cN$ be a von Neumann subalgebra such that the unit elements in $\cM$ and $\cN$ coincide. The function $z^{\dito{\hA}{\cM}}$ corresponding to $\dito{\hA}{\cM}$, the inner daseinisation of $\hA$ to $\cM$, is given by $z^{\dito{\hA}{\cM}}=\zA\circ i=\zA|_{\PM}$. If $\hA$ is bounded from below, then $\dito{\hA}{\cM}=\bjoin\{\hB\in SA(\cM) \mid \hB\leq_s\hA\}\in SA(\cM)$.
\end{proposition}
Hence, $\zA$ encodes all inner daseinisations of $\hA$ to von Neumann subalgebras $\cM\subset\cN$ in a uniform way: each $z^{\dito{\hA}{\cM}}$ is given by simply restricting $\zA$ to the smaller domain $\PM$.

\subsection{Extending the domain} We return to right-continous spectral families and $q$-observable functions. If $\hA\in SA(\cM)\subset SA(\cN)$, then $\hEA_{\cN}=i\circ\hEA_{\cM}:\eR\ra\PN$, which is a composite of two right adjoints (and hence is a right adjoint itself). The corresponding left adjoint is the composite
\begin{align}
			\oA_{\cN}=\oA_{\cM}\circ\deo_{\cM}:\PN\lra\eR
\end{align}
of left adjoints. Hence, one can extend the domain of a $q$-observable function $\oA$ from $\PM$ to $\PM$ by precomposing with $\deo_{\cM}:\PN\ra\PM$.

\begin{remark}
In \cite{DI08b,DI08c,Doe11} and other places, the projection $\dito{\hP}{\cM}$ is called the \emph{inner daseinisation of the projection $\hP$ to $\cM$}, and $\doto{\hP}{\cM}$ is called the \emph{outer daseinisation of $\hP$ to $\cM$}. As mentioned above, $\dito{\hA}{\cM}$ and $\doto{\hA}{\cM}$ are called \emph{inner} and \emph{outer daseinisation of the self-adjoint operator $\hA$ to $\cM$}. These constructions play a central role in the topos approach to quantum theory \cite{DI11}. In the topos approach only bounded self-adjoint operators are considered, so by Lemma \ref{Lem_BoundedFromAboveAndBounded} (resp. Lemma \ref{Lem_BoundedFromBelowAndBounded}) their outer (resp. inner) daseinisations are always bounded self-adjoint operators as well.

The main motivation for the current article was to gain a better mathematical understanding of daseinisation of self-adjoint operators and the underlying operator-theoretic constructions. In standard quantum theory, physical quantities are represented by self-adjoint operators. In the topos approach, physical quantities are represented by certain natural transformations generalising functions from the state space of a system to the real numbers (see \cite{DI08c}). This `function-like' representation of physical quantities is based on daseinisation of self-adjoint operators. As we have seen, to each (bounded) self-adjoint operator $\hA$ there corresponds a $q$-observable function $\oA:\PN\ra\eR$ and a function $\zA:\PN\ra\eR$, and these functions give all the `daseinised' (i.e., approximated) operators $\doto{\hA}{\cM}$ resp. $\dito{\hA}{\cM}$ simply by restriction (Prop. \ref{Prop_OuterDasAndqObsFcts} resp. Prop. \ref{Prop_InnerDasAndzFcts}). In the topos approach, one considers such approximations/restrictions to all \emph{contexts}, i.e., abelian von Neumann subalgebras $V$ of $\cN$ (where $V$ and $\cN$ have the same unit element). Physically, each abelian subalgebra represents one classical perspective on the quantum system.
\end{remark}

If $\mc P$ is some complete sublattice of $\PN$, but not necessarily the lattice of projections of a von Neumann subalgebra $\cM$, then all arguments about restrictions made so far generalise straightforwardly. For example, one may consider a unital $C^*$-algebra $\cA$ and define open projections in its enveloping von Neumann algebra $\cA''$ as those projections that are suprema of nets of positive elements in $\cA$. The open projections form a complete sublattice of the lattice of all projections in $\cA''$, and self-adjoint operators whose spectral families consist only of open projections can be interpreted as analogues of continuous functions \cite{Ake70}, or semicontinuous functions \cite{Bro88}. For a lattice-theoretic treatment of the latter case, see \cite{Com06}.

\subsection{Algebraic structure}
The order-isomorphism \eq{E_SANleqs=QOleq} between the self-adjoint operators affiliated with $\cN$ and their $q$-observable functions, sending $\hA$ to $\oA$, is not a linear map. Counterexamples are easy to find: consider $\hA=\hat 1=\hQ+(\hat 1-\hQ)$, where $\hQ$ is some non-zero projection. We have, for all $\hP\in\PzN$,
\begin{equation}
			o^{\hat 1}(\hP)=\inf\{r\in\eR \mid \hE^{\hat 1}_r\geq P\}=1.
\end{equation}
Assume that $\hP$ is a projection such that $\hP\nleq\hQ$ and $\hP\nleq\hat 1-\hQ$, then 
\begin{align}
			o^{\hQ}(\hP)+o^{\hat 1-\hQ}(\hP) &=\inf\{r\in\eR \mid \hE^{\hQ}_r\geq\hP\}+\inf\{s\in\eR \mid E^{\hat 1-\hQ}_s\geq\hP\}\\
			&= 1+1 = 2,
\end{align}
so $o^{\hQ+(\hat 1-\hQ)}\neq o^{\hQ}+o^{\hat 1-\hQ}$.

Yet, a limited amount of algebraic structure is preserved by the map $\hA\mt\oA$, as we will now show. Let $\hA$ be a self-adjoint operator affiliated with a von Neumann algebra $\cN$, and let $f:\eR\ra\eR$ be a join-preserving function. By the definition of the spectral calculus (see e.g. \cite{KR83+86}), roughly speaking the map $\hA\mt f(\hA)$ shifts the spectral values of $\hA$, but leaves the spectral projections unchanged (apart from generating degeneracies if $f$ is not injective). We will make use of this in the proof of the following result:

\begin{proposition}			\label{P_Ef(A)}
Let $\hA\in SA(\cN)$, let $f:\eR\ra\eR$ be a join-preserving function, and let $g:\eR\ra\eR$ be the right adjoint of $f$. For all $r\in\eR$,
\begin{equation}
			\hE^{f(\hA)}_r=\hEA_{g(r)}.
\end{equation}
\end{proposition}

\begin{proof}
Let $\cB(\bbR)$ denote the Borel subsets of $\bbR$, let $e^{f(\hA)}:\cB(\bbR)\ra\PN$ be the spectral measure of $f(\hA)$, and let $\eA:\cB(\bbR)\ra\PN$ be the spectral measure of $\hA$. Since $f$ is monotone, it shifts the spectral values of $\hA$ in an order-preserving way, hence we have
\begin{equation}
			\hE^{f(\hA)}_r=e^{f(\hA)}((-\infty,r])=\eA((-\infty,x])
\end{equation}
for some $x\in\eR$ that only depends on $r$, so $x=h(r)$ for some function $h:\eR\ra\eR$. If $f$ happens to be injective, we can take $h=f^{-1}$ and thus $h(r)=f^{-1}(r)$. In general, $h(r)$ must be chosen larger than all those $s\in\eR$ for which $f(s)\leq r$ holds (but not larger), so
\begin{equation}
			h(r) = \sup\{s\in\eR \mid f(s)\leq r\},
\end{equation}
that is, $h=g$, the right adjoint of $f$. Then, by the fact that the counit is smaller than the identity (see Rem. \ref{Rem_UnitCounitComposite}), $f(g(r)) \leq r$, and $f(t)>r$ for any $t>g(r)$. We obtain
\begin{equation}
			\hE^{f(\hA)}_r=e^{f(\hA)}((-\infty,r])=e^{\hA}((-\infty,g(r)])=\hEA_{g(r)}.
\end{equation}
\end{proof}

Note that the sort of functions acting on self-adjoint operators that is considered here is characterised by an order-theoretic property, viz. join-preservation. Topologically, this can be understood as functions that are lower semicontinuous. We remark that $\hEA\circ g:\eR\ra\eR$ is meet-preserving, but neither $\bmeet_{r\in\bbR}(\hEA\circ g)(r)=\bmeet_{r\in\bbR}\hEA_{g(r)}=\hat 0$ nor $\bjoin_{r\in\bbR}\hEA_{g(r)}=\hat 1$ need to hold in general. (E.g. pick $g(r)=r_0$ for all $r\in\eR\backslash\{\infty\}$, where $r_0$ is such that $\hat 0<\EA_{r_0}<\hat 1$, and $g(\infty)=\infty$.)

We can now express the $q$-observable function of $f(\hA)$ in terms of the $q$-observable function of $\hA$:

\begin{theorem}			\label{Thm_ofA=foA}
Let $\hA$ be a self-adjoint operator affiliated with a von Neumann algebra $\cN$, and let $f:\eR\ra\eR$ be a join-preserving function. Then
\begin{equation}
			o^{f(\hA)}=f(\oA).
\end{equation}
\end{theorem}

\begin{proof}
Let $g:\eR\ra\eR$ denote the right adjoint of $f$. By definition, $\oA(\hP)=\inf\{t\in\eR \mid \hEA_t\geq\hP\}$, so
\begin{equation}			\label{IE_EAoA}
			\hEA_{g(r)}\geq\hP\quad\Longleftrightarrow\quad g(r)\geq\oA(\hP).
\end{equation}
Moreover, $f:\eR\ra\eR$ is the left adjoint of $g:\eR\ra\eR$, so
\begin{equation}			\label{IE_Adj}
			g(r)\geq z\quad\Longleftrightarrow\quad r\geq f(z).
\end{equation}
We obtain, for all $\hP\in\PN$,
\begin{align}
			o^{f(\hA)}(\hP) &=\inf\{r\in\eR \mid \hE^{f(\hA)}_r\geq\hP\}\\
			&\overset{\text{Prop. }\ref{P_Ef(A)}}{=}\inf\{r\in\eR \mid \hEA_{g(r)}\geq\hP\}\\
			&\overset{\eq{IE_EAoA}}{=}\inf\{r\in\eR \mid g(r)\geq\oA(\hP)\}\\
			&\overset{\eq{IE_Adj}}{=}\inf\{r\in\eR \mid r\geq f(\oA(\hP))\}\\
			&=f(\oA(\hP)).
\end{align}
\end{proof}

\begin{corollary}
Let $t\in\bbR$ and $\hA\in SA(\cN)$, then $o^{\hA+t\hat 1}=\oA+t$. If $s\in\bbR^+$ is a positive real number, then $o^{s\hA}=s\oA$.
\end{corollary}

\begin{proof}
The function
\begin{align}
			f_1:\eR &\lra \eR\\
			r &\lmt r+t,
\end{align}
where $-\infty+t=-\infty$ and $\infty+t=\infty$, preserves suprema. Hence, by Thm. \ref{Thm_ofA=foA}, it holds that
\begin{equation}
			o^{\hA+t\hat 1}=o^{f_1(\hA)}=f_1(\oA)=\oA+t.
\end{equation}
Similarly,
\begin{align}
			f_2:\eR &\lra \eR\\
			r &\lmt sr
\end{align}
preserves suprema, so $o^{s\hA}=s\oA$.
\end{proof}

\section{$q$-antonymous functions}			\label{Sec_qAntonymousFunctions}
It is natural to ask how a $q$-observable function $\oA$ behaves when multiplied by a \emph{negative} real number. The essential issue of course is the behaviour under multiplication by $-1$, which we will consider now. This will lead us to the definition of a second sort of functions associated with self-adjoint operators, besides the $q$-observable functions. 

This new sort of functions, called \emph{$q$-antonymous functions}, is closely related to both $q$-observable functions and the functions of the form $\zA$ that were considered in subsection \ref{Subsec_InnerDasAndzAFcts}.


Let $\hA$ be a self-adjoint operator affiliated with a von Neumann algebra $\cN$, and let $\oA$ be its $q$-observable function. Then, for all $\hP\in\PN$,
\begin{align}
			-\oA(\hP) &= -\inf\{r\in\eR \mid \hP\leq\hEA_r\}\\
			&= \sup\{-r\in\eR \mid \hP\leq\hEA_r\}\\
			&= \sup\{r\in\eR \mid \hP\leq\hEA_{-r}\}.			\label{Eq_-oA}
\end{align}

Let $F^{-\hA}=(\hF^{-\hA}_r)_{r\in\eR}$ be the unique left-continuous extended spectral family of $-\hA$. As is well known,
\begin{equation}
			\forall r\in\eR: \hF^{-\hA}_r = \hat 1-\hEA_{-r}.
\end{equation}
Hence, we obtain from equation \eq{Eq_-oA}
\begin{equation}
			-\oA(\hP) = \sup\{r\in\eR \mid \hP\leq\hat 1-\hF^{-\hA}_r\}.
\end{equation}
As one may have expected, $-\oA$ relates to $-\hA$ (but $-\oA$ is not a $q$-observable function).

\begin{definition}
Let $\hA$ be a self-adjoint operator affiliated with a von Neumann algebra $\cN$. The function
\begin{align}
			\aA: \PN &\lra \eR\\			\nonumber
			\hP &\lmt \sup\{r\in\eR \mid \hP\leq\hat 1-\hFA_r\}
\end{align}
is called the \emph{$q$-antonymous function associated with $\hA$}. The set of $q$-antonymous functions of self-adjoint operators affiliated with $\cN$ is denoted $QA(\PN,\eR)$.
\end{definition}

We already saw that there is a bijection
\begin{align}			\label{Def_BijectionQO=-QA}
			n: QO(\PN,\eR) &\lra QA(\PN,\eR)\\			\nonumber
			\oA &\lmt -\oA = a^{-\hA}
\end{align}
between the sets of $q$-observable and $q$-antonymous functions associated with a von Neumann algebra $\cN$. Thm. \ref{Thm_SA(N)=QO} implies that there is a bijection between $QA(\PN,\eR)$ and $SA(\cN)$, the self-adjoint operators affiliated with $\cN$. Hence,  $q$-antonymous functions provide another representation of self-adjoint operators by real-valued functions.

\begin{proposition}
There is an order-isomorphism
\begin{align}
			\gamma: SA(\cN) &\lra QA(\PN,\eR)\\			\nonumber
			\hA &\lmt \aA.
\end{align}
\end{proposition}

\begin{proof}
This can be proven in a similar way as Prop. \ref{Prop_OrderIso}, or using Prop. \ref{Prop_OrderIso}, more directly as follows: for all $\hA,\hB\in SA(\cN)$,
\begin{align}
			\hA\leq_s\hB &\Longleftrightarrow -\hA\geq_s -\hB\\
			&\Longleftrightarrow o^{-\hA}\geq o^{-\hB}\\
			&\Longleftrightarrow -o^{-\hA}\leq -o^{-\hB}\\
			&\Longleftrightarrow \aA\leq a^{\hB}.
\end{align}
\end{proof}

Clearly, we have, for all $\hA$ affiliated with $\cN$:
\begin{itemize}
	\item [(a)] $\aA:\PN\ra\eR$ is antitone (order-reversing),
	\item [(b)] $\im\aA=-\im o^{-\hA}=-(\sp(-\hA)\cup\{-\infty\})=\sp\hA\cup\{\infty\}$. In particular, $\aA(\hat 0)=\infty$ and $\aA(\hat 1)=\inf\sp\hA$, which is a minimum if $\sp\hA$ is bounded from below, and $-\infty$ otherwise.
\end{itemize}

Moreover, we have:
\begin{lemma}
Let $\hA$ be a self-adjoint operator affiliated with a von Neumann algebra $\cN$, and let $\oA$ and $\aA$ be its $q$-observable resp. $q$-antonymous function. Then
\begin{equation}
			\forall \hP\in\PN\backslash\{\hat 0,\hat 1\}: \aA(\hP)\leq\oA(\hP).
\end{equation}
\end{lemma}

\begin{proof}
Let $\EA$ denote the right-continuous spectral family of $\hA$ and $\FA$ the left-continuous one. Observe that
\begin{align}
			\aA(\hP) &= \sup\{r\in\eR \mid \hP\leq\hat 1-\hFA_r\}\\
			&= \sup\{r\in\eR \mid \hP\leq\hat 1-\hEA_r\}.
\end{align}

Let $r\in\eR$ be such that $\hP\leq\hEA_r$, then $\hP\nless\hat 1-\hEA_r$, so $r>\aA(\hP)$. Since
\begin{equation}
			\oA(\hP)=\inf\{r\in\eR \mid \hP\leq\hEA_r\},
\end{equation}
it follows that $\oA(\hP)\geq\aA(\hP)$.
\end{proof}

Antonymous functions were first introduced by one of us (AD) in a similar form in \cite{Doe05}, where some of their properties were proven. De Groote was the first to point out \cite{deG05} that $-\oA=a^{-\hA}$ (for similar functions).

\begin{lemma}
Let $\hA$ be a self-adjoint operator affiliated with $\cN$, and let $\zA$ be the right adjoint of the left-continuous extended spectral family $\FA$ (see subsection \ref{Subsec_InnerDasAndzAFcts}). Then, for all $\hP\in\PN$, $\aA(\hP)=\zA(\hat 1-\hP)$.
\end{lemma}

\begin{proof}
The adjoint functor theorem for posets gives the concrete form of $\zA$,
\begin{align}
			\zA:\PN &\lra \eR\\			\nonumber
			\hP &\lmt \sup\{r\in\eR \mid \hFA_r\leq\hP\}.
\end{align}
Since $\hat 1-\hFA_r\geq\hP$ if and only if $\hFA_r\leq\hat 1-\hP$, we obtain $\aA(\hP)=\zA(\hat 1-\hP)$.
\end{proof}

Let
\begin{align}
			c:\PN &\lra \PN\\			\nonumber
			\hP &\lmt \hat 1-\hP.
\end{align}
The previous lemma shows that $\aA=\zA\circ c$. Since $\zA$ is a right adjoint, it preserves meets, and $c$ maps joins in $\PN$ to meets in $\PN$, so $\aA$ maps joins in $\PN$ to meets in $\eR$, i.e., to infima. $\aA$ has an adjoint that maps joins in $\eR$, i.e., suprema, to meets in $\PN$, but we do not need this here. It is straightforward to show the following analogue of Prop. \ref{Prop_InnerDasAndzFcts}:

\begin{proposition}			\label{Prop_InnerDasAndAntonFcts}
Let $\cN$ be a von Neumann algebra, let $\hA\in SA(\cM)$, and let $\cM\subset\cN$ be a von Neumann subalgebra such that the unit elements in $\cM$ and $\cN$ coincide. The function $a^{\dito{\hA}{\cM}}$ corresponding to $\dito{\hA}{\cM}$, the inner daseinisation of $\hA$ to $\cM$, is given by $a^{\dito{\hA}{\cM}}=\aA\circ i=\aA|_{\PM}$. If $\hA$ is bounded from below, then $\dito{\hA}{\cM}=\bjoin\{\hB\in SA(\cM) \mid \hB\leq_s\hA\}\in SA(\cM)$.
\end{proposition}

Hence, just like the function $\zA$, the $q$-antonymous function $\aA$ encodes all inner daseinisations of $\hA$ to von Neumann subalgebras $\cM\subset\cN$ in a uniform way: each $a^{\dito{\hA}{\cM}}$ is given by simply restricting $\aA$ to the smaller domain $\PM$. The advantage of using $\aA$ is equation \eq{Def_BijectionQO=-QA}, $\aA=-o^{-\hA}$, which relates $q$-observable and $q$-antonymous functions, and clarifies how $q$-observable functions behave under multiplication by $-1$.

\section{Physical interpretation and outlook}			\label{Sec_PhysInterpret}
We have shown that to each self-adjoint operator $\hA$ affiliated with a von Neumann algebra $\cN$, one can associate a $q$-observable function $\oA:\PN\ra\eR$, which is the left adjoint of the (extended) spectral family $\EA:\eR\ra\PN$ of $\hA$. Conversely, by Thm. \ref{Thm_SA(N)=QO} each self-adjoint operator affiliated with $\cN$ arises from an abstract $q$-observable function, i.e., a join-preserving function $o:\PN\ra\eR$ such that
\begin{itemize}
	\item [(a)] $o(\hP)>-\infty$ for all $\hP>\hat 0$ (which characterises a \emph{weak} $q$-observable function),
	\item [(b)] there is a family $(\hP_i)_{i\in I}\subseteq\PN$ with $\bjoin_{i\in I}\hP_i=\hat 1$ such that $o(\hP_i)\lneq\infty$ for all $i\in I$.
\end{itemize}
The $q$-observable function corresponding to $\hA$ is denoted $\oA$. It is easy to show that the image of $\oA$ on non-zero projections is the spectrum of $\hA$ (see Lemma \ref{Lem_im(oA)=spA}). If $\hA$ is bounded from above, then the family in condition (b) can be chosen to be $\{\hat 1\}$ trivially.

This gives a new, non-standard way of representing physical quantities in quantum theory, which are usually described by self-adjoint operators in von Neumann algebras, or -- for physical quantities like position and momentum which have unbounded spectra -- by self-adjoint operators affiliated with a von Neumann algebra.

The spectral order makes $SA(\cN)$, the set of self-adjoint operators affiliated with $\cN$, a conditionally complete lattice $(SA(\cN),\leq_s)$. We saw in Prop. \ref{Prop_OrderIso} that this lattice is isomorphic to the conditionally complete lattice $(QO(\PN,\eR),\leq)$ of $q$-observable functions with the pointwise order. This shows that $q$-observable functions faithfully represent self-adjoint operators and their order structure, as provided by the spectral order (which differs from the linear order if and only if the von Neumann algebra is nonabelian).

A limited `functional calculus' applies to $q$-observable functions: as long as we consider join-preserving (and hence monotone) functions $f:\eR\ra\eR$, we have $o^{f(\hA)}=f(\oA)$, see Thm. \ref{Thm_ofA=foA}. Physically, such functions $f$ can be interpreted as rescalings of the spectrum of a physical quantity. For such rescalings, it makes sense to preserve the order on $\eR$: if measurements of a physical quantity produce outcomes $a$ and $b$ (which are real numbers in the spectrum of the self-adjoint operator $\hA$) such that $a<b$, and afterwards some rescaling $f:\eR\ra\eR$ is applied, then it is sensible to demand $f(a)\leq f(b)$. In simple cases, such an order-preserving rescaling corresponds to a change of units, e.g. from $^\circ C$ to $^\circ F$.

The question how $q$-observable functions behave under multiplication by $-1$ led us to the definition of $q$-antonymous functions, which correspond bijectively to $q$-observable functions by $-\oA=a^{-\hA}$. This also implies that $q$-antonymous functions correspond bijectively to self-adjoint operators and that $\hA\leq_s\hB$ is equivalent to $\aA\leq a^{\hB}$ in the pointwise order.

The initial motivation for this article was twofold: on the one hand, we wanted to elucidate some of de Groote's results \cite{deG01,deG05,deG07} by providing an order-theoretic perspective and by using Galois connections systematically, while on the other hand, we aimed to gain a deeper understanding (and a generalisation to unbounded operators) of the daseinisation maps from the topos approach to quantum theory \cite{DI08b,DI08c,Doe11}.


Prop. \ref{Prop_OuterDasAndqObsFcts} shows that $\oA|_{\PM}=o^{\doto{\hA}{\cM}}$, that is, the (weak) $q$-observable function of the outer daseinisation of $\hA$ to $\cM$ is given by restriction of the $q$-observable function of $\hA$ to the smaller domain $\PM$. Analogously, we have $\aA|_{\PM}=a^{\dito{\hA}{\cM}}$, that is, the (weak) $q$-antonymous function of the inner daseinisation of $\hA$ to $\cM$ is given by restriction of the $q$-antonymous function of $\hA$ to the smaller domain $\PM$; see Prop. \ref{Prop_InnerDasAndAntonFcts}.

Physically, the inner and outer daseinisations of a bounded self-adjoint operator $\hA$ to a von Neumann subalgebra $\cM\subset\cN$ can be seen as approximations of (the physical quantity described by) $\hA$ to the subalgebra $\cM$. Outer daseinisation is approximation `from above', while inner daseinisation is approximation `from below', both with respect to the spectral order. The more familiar linear order, in which $\hA\leq\hB$ if and only if $\hB-\hA$ is positive, does not give a lattice structure on the self-adjoint operators. (In fact, Kadison famously called $(\cN_{\sa},\leq)$ an \emph{anti-lattice} \cite{Kad51}.) Hence, the approximations from above and from below could not be defined with respect to the linear order, since the relevant meets and joins do not exist in general. Moreover, the approximation with respect to the spectral order guarantees \cite{DI08c,DI11} that 
\begin{equation}
			\sp(\doto{\hA}{\cM})\subseteq\sp\hA,\qquad \sp(\dito{\hA}{\cM})\subseteq\sp\hA,
\end{equation}
which is sensible physically: the approximations $\doto{\hA}{\cM}$ and $\dito{\hA}{\cM}$ have a spectrum that is a subset of the operator $\hA$. In physical terms, the approximations with respect to the spectral order introduce degeneracy, but they do not shift the spectral values.

Often, $\cM$ is taken to be abelian and thus represents a measurement context, since commuting self-adjoint operators correspond to compatible, i.e., co-measurable physical quantities. The representation of physical quantities in the topos approach to quantum theory is based on this idea of approximation \cite{DI08c,DI11}, with $\cM$ varying over the abelian von Neumann subalgebras of $\cN$ that share the unit element with $\cN$.

\textbf{Outlook.} In the second paper \cite{DoeDew12b}, entitled ``Self-adjoint Operators as Functions II: Quantum Probability'', we will interpret $q$-observable functions in the light of quantum probability theory. It will be shown that $\oA:\PN\ra\eR$ is the generalised quantile function of the projection-valued spectral measure $\eA:\cB(\bbR)\ra\PN$ given by a self-adjoint operator $\hA$, where $\cB(\bbR)$ denotes the Borel subsets of the real line. In this perspective, the spectral family $\EA:\bbR\ra\PN$ takes the role of the cumulative distribution function of the measure $\eA$.

Using some results from the topos approach to quantum theory, it will be shown that the so-called \emph{spectral presheaf $\Sig$} of a von Neumann algebra can serve as a joint sample space for \emph{all} quantum observables (seen as random variables). This is possible despite no-go theorems like the Kochen-Specker theorem, since $\Sig$ is a generalised set -- in fact, an object in a presheaf topos -- with a complete bi-Heyting algebra of measurable subobjects, generalising the classical situation of a set with a $\sigma$-algebra of measurable subsets. We will show how the Born rule arises in this new formulation.

\textbf{Acknowledgements.} We thank Prakash Panangaden, Chris Isham, Rui Soares Barbosa and Daniel Marsden for discussions, and we thank Massoud Amini for bringing Comman's article to our attention. We also thank the anonymous referee whose careful reading and suggestions led to a number of improvements. We are grateful to Masanao Ozawa, Izumi Ojima, Mikl\'os R\'edei, Pekka Lahti and John Harding for their kind interest and for further suggestions. Some of these suggestions were incorporated, while others will be developed in a future paper.


\begin{thebibliography}{9}                                                    

\bibitem{Ake70} C.A.~Akemann, ``Left Ideal Structure of $C^*$-Algebras'', J. Funct. Anal. \textbf{6}, 305--317 (1970).

\bibitem{Awo10} S.~Awodey, \textit{Category Theory}, Second Edition, Oxford Logic Guides \textbf{52}, Oxford University Press, Oxford (2010).
                         
\bibitem{Com06} H.~Comman, ``Upper Regularization for Extended Self-Adjoint Operators'', Journal of Operator Theory \textbf{55}, no. 1, 91--116  (2006).

\bibitem{Bir67} G.~Birkhoff, ``Lattice Theory'', third edition, American Mathematical Society (1967).

\bibitem{Bro88} L.G.~Brown, ``Semicontinuity and multipliers of $C^*$-algebras'', \emph{Can. J. Math.} \textbf{XL}, No. 4, 865--988 (1988).
 
\bibitem{DavPri02} B.A.~Davey, H.A.~Priestley, \textit{Introduction to Lattices and Order}, second edition, Cambridge University Press, Cambridge (2002).

\bibitem{Doe05} A.~D\"oring, ``Observables as functions: Antonymous functions'', arXiv:quant-ph/0510102 (2005).

\bibitem{Doe11} A.~D\"oring, ``The Physical Interpretation of Daseinisation'',in \textit{Deep Beauty}, ed. Hans Halvorson, Cambridge University Press, Cambridge, 207--238 (2011).

\bibitem{DoeDew12b} A.~D\"oring, B. Dewitt, ``Self-adjoint Operators as Functions II: Quantum Probability'', submitted to CMP, related online version: arXiv:1210.5747 (2012).

\bibitem{DLM12} A.~D\"oring, R.~Lal, D.~Marsden, ``Order-Theoretic Aspects of Daseinisation and Intuitionistic Quantum Logic'', in preparation (2012).

\bibitem {DI08a} A.~D\"{o}ring, C.J. Isham, ``A topos foundation for theories of physics: I. Formal languages for physics'', \emph{J. Math. Phys.} {\bf 49}, 053515 (2008).

\bibitem{DI08b} A.~D\"oring, C.J.~Isham, ``A topos foundation for theories of physics: II. Daseinisation and the liberation of quantum theory'', \emph{J. Math. Phys.} {\bf 49}, 053516 (2008).

\bibitem{DI08c} A.~D\"{o}ring, C.J. Isham, ``A topos foundation for theories of physics: III. Quantum theory and the representation of physical quantities with arrows \mbox{$\breve{A}:\underline{\Sigma}\ra\underline{\bbR^{\succeq}}$}'', \emph{J. Math. Phys.} {\bf 49}, 053517 (2008).

\bibitem{DI08d} A.~D\"{o}ring, C.J. Isham, ``A topos foundation for theories of physics: IV. Categories of systems'', \emph{J. Math. Phys.} {\bf 49}, 053518 (2008).

\bibitem{DI11}  A.~D\"{o}ring, C.J. Isham, ```What is a Thing?': Topos Theory in the Foundations of Physics'', in \emph{New Structures for Physics}, ed. B.~Coecke, Lecture Notes in Physics \textbf{813}, Springer, Heidelberg, Dordrecht, London, New York, 753--937 (2011).

\bibitem{deG01} H.F.~de Groote, ``Quantum Sheaves: An outline of results'', arXiv:math-ph/0110035 (2001).

\bibitem{deG04} H.F.~de Groote, ``On a canonical lattice structure on the effect algebra of a von Neumann algebra'', arXiv:math-ph/0410018v2 (2004).

\bibitem{deG05} H.F.~de Groote, ``Observables II: Quantum Observables'', arXiv:math-ph/0509075 (2005).

\bibitem{deG07} H.F.~de Groote, ``Observables IV: The Presheaf Perspective'', arXiv:0708.0677 [math-ph] (2007).

\bibitem{HLS09} C.~Heunen, N.P.~Landsman, B.~Spitters, ``A topos for algebraic quantum theory'', \emph{Comm.\ Math.\ Phys.} \textbf{291}, 63--110 (2009).

\bibitem{Kad51} R.V.~Kadison, ``Order Properties of Bounded Self-adjoint Operators'', \emph{Proc. of the American Mathematical Society} \textbf{2}, No. 3, pp. 505-510 (1951).

\bibitem{KR83+86} R.V.~Kadison, J.R.~Ringrose, \textit{Fundamentals of the Theory of Operator Algebras}, vols. I\&II, Academic Press, New York, London (1983).

\bibitem{McL98} S.~MacLane, \emph{Categories for the Working Mathematician}, second edition, Springer, New York, Berlin, Heidelberg (1998).

\bibitem{Ols71} M. P.~Olson, ``The selfadjoint operators of a von Neumann algebra form a conditionally complete lattice'', {\em Proc. of the AMS} {\bf 28}, 537--544 (1971).

\bibitem{PlaSto12} A.~Planeta, J.~Stochel, ``Spectral order for unbounded operators'', \textit{J. Math. Anal. Appl.} \textbf{389}, 1029--1045 (2012).

\bibitem{Red09} M.~Redei, ``The Birkhoff--von Neumann Concept of Quantum Logic'', in \textit{Handbook of Quantum Logic and Quantum Structures}, eds. K.~Engesser, D.M.~Gabbay, D.~Lehmann, North-Holland, Amsterdam (2009).

\bibitem{Wol10} S.~Wolters, ``A Comparison of Two Topos-Theoretic Approaches to Quantum Theory'', arXiv:1010.2031v2 (2010; version 2 from 3. August 2011).

\bibitem{Tak79+02} M.~Takesaki, \emph{Theory of Operator Algebras}, vols. I-III, Springer, Berlin, Heidelberg, New York (1979/2002).

\end{thebibliography}
\end{document}